\documentclass[12pt]{amsart}
\usepackage{amsmath,amscd,amsthm,amsfonts, amssymb,amsxtra, mathrsfs, amsrefs}
\usepackage{accents}
\usepackage[us,12hr]{datetime}
\usepackage[all]{xy} \SelectTips{cm}{}
\usepackage[dvipsnames]{xcolor}
\usepackage{hyperref}
 \hypersetup{colorlinks=true,citecolor=blue,linkcolor=olive}

   \usepackage{graphicx}
  \usepackage[margin=1.75in]{geometry}
  \usepackage[pagewise, displaymath, mathlines]{lineno}

   \usepackage[shortlabels]{enumitem}

\CDat
\newtheorem{thm}{Theorem}[section]  
\newtheorem*{un-no-thm}{Theorem}
\newtheorem{cor}[thm]{Corollary}     
\newtheorem{lem}[thm]{Lemma}         
\newtheorem{prop}[thm]{Proposition}  
\newtheorem{add}[thm]{Addendum}
\newtheorem{conjecture}[thm]{Conjecture}
\newtheorem{bigthm}{Theorem}

\theoremstyle{definition} 
\newtheorem{defn}[thm]{Definition}   

\theoremstyle{definition}
\newtheorem*{prob}{Problem}   

\theoremstyle{definition}
\theoremstyle{remark}
\newtheorem{rem}[thm]{Remark}

\newtheorem{hypo}[thm]{Hypothesis}
\newtheorem{notation}[thm]{Notation}
\newtheorem*{acks}{Acknowledgements}

\newtheorem*{out}{Outline}

\newtheorem*{intro-rem}{Remark}
\newtheorem*{intro-rems}{Remarks}

\newtheorem{ex}[thm]{Example}

\DeclareMathOperator{\spec}{Spec}
\DeclareMathOperator{\symm}{sym}

\DeclareMathOperator{\self}{End}

\DeclareMathOperator{\U}{U}
\DeclareMathOperator{\Or}{O}
\DeclareMathOperator{\SO}{SO}
\DeclareMathOperator{\PSO}{PSO}

\DeclareMathOperator{\GL}{GL}

\DeclareMathOperator{\PGL}{PGL}

\DeclareMathOperator{\tr}{tr}

\def\:{\colon\!}
\def\cal{\mathcal}
\def\Bbb{\mathbb}
\def\scr{\mathscr}

\begin{document}

\title{On the variety of X-states}
\date{\today} 
\author[L.~Candelori, V.~Y.~Chernyak, and J.~R.~Klein]{
Luca Candelori,$^{\!\! 1}$~Vladimir Y. Chernyak,$^{\!\! 1,2}$~John R. Klein,$^{\!\! 1}$
}

\thanks{\hskip -.17in $^{1}$\textsc{\tiny Department of Mathematics, Wayne State University, Detroit, MI 48201, USA}\\
$^{2}$\textsc{\tiny Department of Chemistry, Wayne State University, Detroit, MI 48201, USA}}

\subjclass[2020]{Primary: 81P40, 81P42, 13A50, Secondary: 14L24}

\begin{abstract}  We introduce the notion of an X-state on $n$-qubits. After taking
the Zariski closure of the set of X-states in the space of all mixed states, we obtain a complex algebraic variety $\scr X$
that is equipped with the action of the Lie group of local symmetries $G$. We show
that the field of $G$-invariant rational functions on $\scr X$ is purely transcendental over
the complex numbers of degree $2^{2n-1}-n-1$. 
 \end{abstract}

\maketitle
\setcounter{tocdepth}{1}
{\color{olive}
\tableofcontents}
\addcontentsline{file}{sec_unit}{entry}

\section{Introduction} \label{sec:intro}
In recent years, the term ``X-states'' has come to refer to a certain class of  mixed states in 
two qubit quantum systems whose density matrices have the form
\begin{equation*}
\begin{pmatrix} 
\rho_{11} & 0  & 0 & \rho_{14} \\
0 & \rho_{22} &\rho_{23} & 0 \\
0& \rho_{32}& \rho_{33} & 0 \\
\rho_{41} & 0 & 0 & \rho_{44} 
\end{pmatrix}
\end{equation*}
\cite{X-Eberly}, \cite{Rau_2009}, \cite{MENDONCA201479}, \cite{X-states-Gerdt-et-al}.
This class of mixed states arises in the study of quantum entanglement. 
Various states of interest, such as maximally entangled Bell states, Werner states, and maximally isotropic states,
are special cases of X-states.

The above definition of X-states depends on the choice of an ordered basis 
  for the individual qubits. For this reason, the above notion is {\it not} invariant under quantum evolution.
  To correct for this defect, the domain of X-states is usually extended to all mixed states of two qubits that can be obtained by 
  quantum evolution of an X-state. However, this ``extended" definition  does not have the structure of a 
  manifold or an algebraic variety.
  In other words, we cannot study the geometry of X-states directly. 
  
In this article we 
introduce the concept of an {\it X-state on $n$-qubits} for any positive integer $n$. The set of such X-states
is invariant with respect to the Lie group $G$ of local symmetries of the $n$-qubit system.\footnote{As we are working over the complex numbers, in our setting 
$G$ is isomorphic to a cartesian product of $n$-copies of the special orthogonal group $\SO_3(\Bbb C)$.}
To make this new notion accessible to the tools of algebraic geometry,
we take the Zariski closure of the set of X-states on $n$-qubits inside the affine space of all mixed states on $n$-qubits.
This results in an algebraic variety $\scr X_n$ over the complex numbers which remains invariant with respect to
 the action of the group $G$.  In what follows, when $n$ is understood, we drop the subscript from the notation
 and refer to $\scr X_n$ as $\scr X$.

We investigate the variety $\scr X$ together with its $G$-action using the techniques of geometric invariant theory in the affine setting.
A  rational function  $\scr X \to \Bbb C$ which is invariant with respect to $G$
may be viewed as a detailed {\it measure of entanglement} \cite[\S XV]{Horodecki}.
We will  provide a description of the algebraic structure of such measures of entanglement on $\scr X$.
 Our main result is that 
   the field of $G$-invariant rational functions  $\scr X \to \Bbb C$ 
is freely generated by  $2^{2n-1} - n-1$ algebraically independent elements:

\begin{bigthm} \label{bigthm:main} The field $\Bbb C(\scr X)^G$ is purely transcendental over $\Bbb C$ of degree $2^{2n-1} - n-1$.
\end{bigthm}

\begin{rem} In the language of algebraic geometry, Theorem \ref{bigthm:main}  asserts that the algebraic quotient 
\[
\scr X/\!\!/G = \spec(\Bbb C[X]^G)
\] 
is rational of dimension  $2^{2n-1} - n-1$. 
\end{rem}

\begin{rem} In Appendix \ref{sec:rationality-for-2-qubits},
we exhibit a transcendence basis $p_1,\dots, p_5$ for $\Bbb C(\scr X_2)^G$.
By a different method, one may provide a transcendence basis for $\Bbb C(\scr X_n)^G$
when $n > 2$. For reasons of space, we defer the latter to another paper.
\end{rem}

\begin{rem} Although not investigated here, there is an analog of Theorem \ref{bigthm:main}
over the real numbers that takes into account the real structure of the complex algebraic variety $\scr X$. 
In the real case, the symmetry group governing quantum evolution
 is the local unitary group. The latter is isomorphic to the $n$-fold cartesian
product of copies of unitary group $\U(2)$.
\end{rem}

\subsection{Background from physics}
X-states naturally appear  in spectroscopy with either quantum or entangled photons \cite{CiracPRL2004}, \cite{KimPRL2004}, \cite{PalmaPRA2004}, \cite{PalmaPRB2004}, \cite{ZouPRA2006}. The paper 
 \cite{X-states-physics} showed that when $n =2$ qubits which are initially separable 
are exposed to an optical pulse of an entangled photon field, then the result is {\it always} a mixed state
which is also an X-state, due to the $\Bbb Z/2$ symmetry of the entangled photon correlations. 
With our definition of X-state on $n$-qubits, the same phenomenon
occurs for higher values of $n$. This is a reason why X-states are important and fundamental for the
fields of quantum information and quantum computing.

The reader may reasonably ask why mixed states are of relevance in such a setting. In the case
of coherent non-linear spectroscopy, the driving field may be viewed as an external classical field,
so that pure states are preserved by the evolution of the system. However, in the case of quantum light spectroscopy
this is no longer the case: one must necessarily deal with the  space of a joint system consisting of qubits and a quantum electromagnetic field. 
Consequently, the states of the qubit system can be considered only as a result of tracing out the quantum electromagnetic field, so that any initially pure state results in a mixed counterpart.

\begin{out} In \S\ref{sec:prelim} we outline some aspects geometric invariant theory in the affine case.
In \S\ref{sec:L-states} we consider the space $\scr L$ of mixed states on $n$-qubits. In \S\ref{sec:X-states}
we define the variety of X-states $\scr X$ and compute its dimension.  In \S\ref{sec:rel-sec} we construct
a rational section $X(\frak B)$ for the action of $G$ on $\scr X$ and identify the Weyl group $W$. In \S\ref{sec:overview} we provide an overview of the proof
of Theorem \ref{bigthm:main}. In \S\ref{sec:aux} we solve the auxiliary problem of constructing an explicit rational
quotient for the left multiplication action  of the group
$\SO_2(\Bbb C)$ on the vector space $M_2(\Bbb C)$ of $2\times 2$ matrices. \S\ref{sec:quotient-by-wprime} and
\S\ref{sec:w-wprime} in essence provide the bulk of the work in establishing that the algebraic quotient
on $X(\frak B)/\!\! /W$ is rational.
In \S\ref{sec:main} we assemble the results of previous sections to prove Theorem \ref{bigthm:main}. 

We also provide two appendices. The first of these, which is used in \S\ref{sec:w-wprime},
 provides a mild generalization of descent along a torsor in the case
of a fibration associated to a trivial vector bundle equipped with $G$-linearization
and the additional structure of a $G$-invariant ``subvariety along the fibers.'' 
 In the second appendix we construct an explicit transcendence
basis for $\Bbb C(\scr X_2)^G$.
\end{out}

\begin{acks}
  The authors are supported by the U.S. Department of Energy, Office of Science, Basic Energy Sciences, under Award Number DE-SC-SC0022134.
  \end{acks}

\section{Preliminaries}  \label{sec:prelim}

\subsection{Varieties}

In what follows, $k$ will denote algebraically closed field of characteristic zero. A {\it $k$-scheme} is a scheme over $\spec(k)$.
A scheme $X$ is {\it affine} it is isomorphic to $\spec(R)$ for some commutative ring $R$.

A {\it $k$-variety} $X$ is a $k$-scheme which is integral (i.e., irreducible and reduced) and of finite type (but not necessarily separated).
When $k$ is understood, we refer to a $k$-variety as a variety. 

If $X$ is a variety, then an irreducible open subscheme  $Y \subset X$ is again a variety. 
If $X$ is a variety, then an irreducible closed subscheme of $Y\subset X$
is a variety if $Y$ is reduced. 
  
 Note that if $Y$ is any closed subset of a $k$-scheme $X$, then there is  a unique closed subscheme $Z\subset X$ such that $Z$ is reduced
and has the property that the underlying underlying space of $Z$ is $Y$.
The scheme $Z$ is called the {\it reduced induced scheme structure} on $Y$. In particular, if $X$ is a variety and $Y$ is irreducible, then $Z$ is a variety.

Two varieties $X$ and $Y$ are {\it birationally equivalent}
if their function fields are isomorphic. Geometrically, this means that
there are Zariski dense open
subvarieties $X_0\subset X$ and $Y_0 \subset Y$ with $X_0$ isomorphic to $Y_0$.

\begin{defn} A variety $X$  is {\it rational} if it is birationally equivalent to an affine space
  $\Bbb A^d$. Equivalently, $X$ is rational if its function field is isomorphic to $k(U_1,\dots,U_d)$, the field
  of rational functions on $d$ indeterminates $U_1,\dots, U_d$.
  \end{defn}
  

 

    \subsection{Invariant functions and algebraic quotients} 
    
   Let $X$ be a variety and let $G$ an affine algebraic group over $k$ acting (algebraically) on $X$.
Recall that $G$ is said to be {\it reductive} if it possesses a faithful semi-simple representation over $k$.
In what follows, we will always assume that $G$ is reductive.

  We write the coordinate ring (i.e., the global sections of the structure sheaf) of $X$ by $k[X]$  
  and the associated function field by $k(X)$.
  We denote the ring of polynomial $G$-invariant functions on $X$ by $k[X]^G$, so $k[X]^{G}$ is a $k$-subalgebra of $k[X]$.
  Similarly, we write the field of rational $G$-invariant functions by $k(X)^{G}$.
  
 \begin{defn} The {\it algebraic quotient}\footnote{This is also known in the literature as the {\it affine quotient}.} of $X$ by $G$ is the affine scheme
 \[
X/\!\!/G :=  \spec(k[X]^G)\, .
  \]
  We define a {\it rational quotient} for the action of $G$ on $X$ to be
  any variety $Y$ which is birationally equivalent to the
  affine quotient. Equivalently, 
  \[
  k(Y) \cong k(X)^G\, .
  \]
  \end{defn}

  \begin{lem}\label{lem:double-quotient1}  
   Suppose $G$, $H$ are reductive groups acting algebraically on $X$ and $Y$ respectively.
      If $X$ and $Y$ are affine varieties, then
      there is an isomorphism
      \[
        (X\times Y)/\!\!/(G\times H)\cong (X/\!\!/ G)\times (Y/\!\!/H)\, .
      \]
   \end{lem}
   
     \begin{proof} This follows immediately from the observation that
      it $X=\spec(R)$ and $Y=\spec(S)$ then  we have 
      \[
      R^{G}\otimes_{k} S^H \cong(R\otimes_{k} S)^{G\times H}\, . \qedhere
      \]
    \end{proof}
   
     Suppose that  $G$ is reductive and acts on an affine variety $X$. Suppose $N \subset G$
  is a Zariski closed normal subgroup. Then the quotient group $G/N$ is defined, is reductive, and is affine 
  \cite[\S16]{Waterhouse}.

     \begin{lem} \label{lem:double-quotient2}
      Suppose $G$ is a reductive group acting on an affine variety $X$
      and suppose that $N\subset G$ is a Zariski closed normal subgroup.
      Then there is an isomorphism of algebraic quotients 
      \[
      X/\!\!/G=(X/\!\!/H)/\!\!/ (G/H)\, .
      \]
    \end{lem}

 \begin{proof}
  This follows from the observation that if $ X = \spec(R)$, then 
  $R^G = (R^{N})^{G/N}$.
 \end{proof}
 
In the case when $G = A \times B$, we have $A\subset G$ is normal with quotient $B$.
 Consequently,
 
   \begin{cor}\label{cor:double-quotient3}  
   The algebraic quotient for $G = A\times B$ acting on $X$ is given by
$
(X/\!\!/A)/\!\!/B
$.
 \end{cor}

  \subsection{Rational sections} Let $G$ act on $X$.
  The {\it normalizer} of a subvariety $\cal S \subset X$ is the subgroup $N\subset G$ 
defined by $\{g\in G\mid g\cal S\subseteq \cal S\}$. The {\it centralizer} is the normal subgroup
$C\subset N$ consisting of those elements of $G$ which fix $\cal S$ pointwise. The {\it Weyl group}
is the quotient $W = N/C$.

  \begin{defn}[{\cite[\S2.8]{PV}}]
    \label{defn:relativeSection}
    We say a subvariety $\cal S\subseteq X$ is a \textit{rational section} (with respect to $G$) if there exists a
   dense open subset $\cal S_0 \subset \cal S$ such that
      \begin{enumerate} 
      \item The Zariski closure of the orbit of $\cal S$ is $X$, i.e., $\overline{G\cal S}=X$.
      \item If $g \in G$ and $g \cal \cal S_0\cap \cal S_0\neq\emptyset$, then $g\in N$.
    \end{enumerate}
  \end{defn}
  It is not hard to see if $\cal S\subseteq X$ is a rational section then the restriction $\mathbb{C}[X]^{G}\to\mathbb{C}[\cal S]^{W}$ is injective.
  The latter induces an field homomorphism $\mathbb{C}(X)^{G}\to\mathbb{C}(\cal S)^{W}$.
 In fact, 
 
  \begin{prop}[{\cite[p.~161]{PV}}] \label{prop:sec-rel}
    Let $\cal S\subseteq X$ be a rational section with respect to $G$.
    Then the restriction homomorphism $k(X)^{G}\to k(\cal S)^{W}$ is an isomorphism.    
  \end{prop}

 \begin{defn}
    \label{def:almostFree}
    The action of $G$ on $X$ is \emph{almost free} if there exists an open dense $G$-invariant subset $U\subseteq X$ such that for every closed point $x \in U$, the 
    stabilizer $G_x \subseteq G$ is trivial. 
  \end{defn}

\subsection{Rationality }
 A useful tool for proving rationality is the following version of descent for vector bundles along a torsor.

\begin{prop}[Descent {\cite{dolgachev}, \cite[2.13]{PV},  \cite{CT-S}}]\label{prop:no-name}
  Let $G$ be a reductive group that acts almost freely on a variety $X$.
    Let $\pi \: E\to X$ be a $G$-linearized vector bundle. Then the algebraic quotient
    $E/\!\!/G$ is birationally equivalent
    to  a vector bundle over the algebraic quotient $X/\!\!/G$.   
    \end{prop}
    
    \begin{rem}     In particular, $E/\!\!/G$ is rational if $X/\!\!/G$ is rational.
    
    Let $d$ be the rank of  $\pi$.
    Then with respect to the assumptions,
    another formulation of Proposition \ref{prop:no-name} is that there is a commutative diagram of varieties with $G$-action
    \[
    \xymatrix{
    U\times \Bbb A^d \ar[dr]_{p_1} \ar[rr]^\cong &&  \pi^{\ast}U\ar[dl]^\pi \\
   &  U
   }
   \]
   in which $U \subset X$ is an open subset on which $G$ acts freely,  $p_1$ is first factor projection,
   $G$ acts trivially on $\Bbb A^d$, and the horizontal map is an isomorphism of vector bundles.
   
Note that the construction of such a diagram is equivalent to finding a basis of equivariant sections of 
$\pi$ along the open set $U$.
\end{rem}

    \begin{cor}[``No-Name Lemma'' {cf.~\cite[cor.~3.8]{CT-S}}] \label{cor:no-name} 
There is an isomorphism of function fields
    \[
 k(E)^G \cong k(\Bbb A^d) \otimes_k k(X)^G\, .
    \]
    \end{cor}

\section{L-states}  \label{sec:L-states}
\subsection{Qubits}  \label{subsec:analytical-scalar}

A {\it qubit} is a complex vector space of dimension two. If $V$ is a qubit, then
we write 
\[
\scr V := \mathfrak{sl}(V)
\] for the Lie algebra of {\it  traceless endomorphisms} of $V$.
The {\it analytical scalar product} $\scr V\otimes \scr V \to \Bbb C$ is given by the Killing form
\[
\langle A,B\rangle := \tfrac{1}{2}\tr(AB)\, ,
\]
in which $\tr$ denotes the trace of an operator.
When $A = B$, we will resort to the notation $\lVert A\rVert^2 := \langle A, A\rangle$.

The {\it vector product} 
\[
[{-},{-}]\: \scr V\otimes \scr V \to \scr V
\]
is  the Lie bracket $[A,B] = AB-BA$.  Then the triple scalar product
\begin{align*}
\Lambda^3 \scr V &@> \cong >> \Bbb C\\
A\wedge B \wedge C &\mapsto \langle [A,B],C\rangle
\end{align*}
equips $\scr V$ with a canonical orientation. In particular, the special orthogonal group 
\[
\SO(\scr V)
\] 
of linear orientation preserving isometries of $\scr V$
is defined.

\begin{rem} If we choose an ordered basis for $V$, 
then $V$ is identified with $\Bbb C^2$ and we acquire an isomorphism of Lie algebras
 $\mathfrak{sl}_2(\Bbb C) \cong \scr V$.  With respect
 to this identification, $\scr V$ inherits positively
oriented ordered  basis of  {\it Pauli matrices}
\[
\sigma_x = \begin{pmatrix}
0& 1\\
1 & 0
\end{pmatrix},
\sigma_y = \begin{pmatrix}
0 & -i\\
i & 0
\end{pmatrix} \! ,
\sigma_z = \begin{pmatrix}
1 & 0\\
0 & -1
\end{pmatrix}\! .
\]
\end{rem} 

\begin{lem} There is a canonical isomorphism
\[
\SO(\scr V) \cong \PGL(V)\, .
\]
\end{lem}

\begin{proof} The group $\GL(V)$ acts by conjugation on
endomorphisms of $V$ and preserves traces. Hence
$\GL(V)$ acts on 
 $\scr V$ and the action factors through $\PGL(V)$. 
Hence, the adjoint to the action $\PGL(V) \times V \to V$ defines a canonical homomorphism $\PGL(V) \to \Or(\scr V)$, where the target
denotes the orthogonal group of $\scr V$. This homomorphism
 is injective and onto the identity component.
\end{proof}

\begin{rem} 
 With respect to the identification $\SO(\scr V) \cong \PGL(V)$, the action of $\PGL(V)$ on $\scr V$ 
 is the standard action by isometries.
 \end{rem}

\subsection{L-states on $n$-qubits}
Given qubits $V_1,\dots,V_n$,
we define  the space of {\it L-states} 
\[
\scr L := \scr L(V_1,\dots,V_n) \subset \self(\otimes_i V_i)
  \]
as the codimension one affine subspace consisting of the trace one endomorphisms
of the tensor product $V_1\otimes \cdots \otimes  V_n$.  Elements
$\rho \in \scr L$ will be referred to as {\it L-states}.\footnote{In the quantum mechanics literature, $\scr L$ is known as
the {\it Liouville space} and $\rho$ is  known as a {\it mixed state} or {\it density operator}.}

Note that as a $\Bbb C$-scheme, $\scr L\subset \Bbb A^{2^n}$  is of codimension one and 
is isomorphic to $\Bbb A^{2^n-1}$.

The complex Lie group 
\begin{equation} \label{eqn:G}
G:= \prod_{i=1}^n \PGL(V_i) \cong \prod_{i=1}^n \SO(\scr V_i) 
\end{equation}
acts by conjugation on endomorphisms, and therefore on $\scr L$.

\subsubsection{The Bloch model for L-states} Let 
\[
\overline{n} = \{1 < \dots < n\}
\]
denote the totally ordered set on $n$ elements.
For a subset
\[
I = \{i_1 < i_2 < \cdots < i_k\} \subset \overline{n}\, ,
\] 
we set
\[
  V_I := V_{i_1} \otimes  \cdots \otimes V_{i_k} \, ,\quad  \scr V_I := \scr V_{i_1} \otimes  \cdots \otimes \scr V_{i_k}   \, .
\]

There is a canonical splitting  $\self(V_i) \cong \Bbb C \oplus \scr V_i$ and which induces s a decomposition
\[
\self(V_{\overline{n}} ) \,\,  \cong\,\,   \bigoplus_{I\subset {\overline{n}}} \scr V_I \, .
\]
Moreover, the condition that an endomorphism be of trace one on the left
corresponds the the element $1 \in \Bbb C = \scr V_\emptyset$ on the right.
This results in a canonical isomorphism
\[
\scr L \cong  \bigoplus_{\emptyset \ne I\subset {\overline{n}}} \scr V_I \, .
\]
\begin{defn} The {\it Bloch model} for L-states on $n$-qubits is
\[
\scr L^b := \bigoplus_{\emptyset \ne I\subset {\overline{n}}} \scr V_I \, .
\]
\end{defn}

Consequently, there is a canonical isomorphism
\begin{align*}
\scr L &@> \cong >> \scr L^b\, ,\\
\rho       &\mapsto \rho^b\, .
\end{align*}

\begin{ex}[2-qubits] 
We have
\[
\scr L^b_2 = \scr V_1 \oplus \scr V_2 \oplus  \scr V_{12} \, ,
\]
where $\scr V_{12} = \scr V_1 \otimes \scr V_2$.
\end{ex}

\begin{defn} Given  $\rho^b \in \scr L^b$, its image  $\rho_I \in \scr V_I$
with respect to the canonical projection $\scr L^b \to \scr V_I$ is the 
{\it $I$-correlation function}
of $\rho$. 

If $|I| = k$, we say that $\rho_I$ is a {\it $k$-point correlation function} of $\rho$.
\end{defn}

\begin{rem} 
The action of $G$ on $\scr L^b$ is componentwise. The action on the component
$\scr V_I$ factors through the projection
$G \to G_I := \prod_{i \in I}\SO(\scr V_i)$.
\end{rem}

  \section{X-states on $n$-qubits}  \label{sec:X-states}

 Recall that an L-state
is a trace one endomorphism $\rho$ of the tensor product
 $V_{\overline{n}}  = V_1 \otimes \cdots \otimes V_n$
 where $V_i$ is a complex vector space of dimension two.
 If we equip each $V_i$ with an ordered basis $\{e_0^i,e_1^i\}$, 
 then the tensor product $V_{\overline{n}} $ inherits a basis. The set of basis elements of the latter
 are in bijection with functions $\overline{n} \to \{0,1\}$, i.e., binary strings of length $n$, in which a function $\phi$ corresponds to
 the basis element
\[
e^1_{\phi(1)} \otimes \cdots \otimes e^n_{\phi(n)} \, .
\]
We may order these basis elements by means of the left lexicographic ordering of functions. 
In physics notation, one may write this basis element as the ket $|\phi(1)\phi(2)\cdots\phi(n)\rangle$.

A basis element is said to be of {\it even/odd parity} if and only if the sum $\sum_i \phi(i)$ is even/odd. The parity
determines a {\it super vector space structure}
\[
V_{\overline{n}}  = V^0_{\overline{n}}  \oplus V^1_{\overline{n}} \,  ,
\]
where $V^0_{\overline{n}} $ and $V^1_{\overline{n}} $ are the spans of the even and odd parity basis elements respectively.
In what follows, recall that a morphism of super vector spaces is a parity-preserving linear transformation.

\begin{defn} An L-state  $\rho \in \scr L$ is said to be a
 {\it  X-state} if  and only if $\rho$
  is a morphism of super vector spaces with respect to some choice
of ordered basis for the individual qubits. 
\end{defn}

We emphasize that in the above definition, the choice of basis sets is permitted to vary.

 \begin{ex}[$n=2$] In the case of 2-qubits $V_1,V_2$, an $L$-state $\rho$ is an X-state iff and only
 if for some choice of ordered basis for $V_1$ and $V_2$, the $(4\times 4)$-matrix corresponding to $\rho$
has the form
\begin{equation*}
\begin{pmatrix} 
\rho_{11} & 0  & 0 & \rho_{14} \\
0 & \rho_{22} &\rho_{23} & 0 \\
0& \rho_{32}& \rho_{33} & 0 \\
\rho_{41} & 0 & 0 & \rho_{44} 
\end{pmatrix}.
\end{equation*}
This description explains the origin of  the terminology ``X-state.'' 
\end{ex}

\subsection{The Bloch model for X-states}

Let $V$ be a qubit with associated Lie algebra $\scr V$.

\begin{defn} 
A  vector $v\in \scr V$ is {\it degenerate} if and only if $\|v\|^2 :=\langle v,v\rangle =0$, otherwise, $v$ is said to be {\it non-degenerate}.
\end{defn}

The set of degenerate vectors for $\scr V$ is a vector subspace $H\subset \scr V$ of codimension one.
Let 
\[
\Bbb P(\scr V) = (\scr V\setminus \{0\})/\Bbb G_m
\]
be the projectivization of $\scr V$ where $\Bbb G_m = \Bbb C^{\times} $ is the multiplicative group of the field  $\Bbb C$.
A point $L \in \Bbb P(\scr V)$  is {\it non-degenerate} if some non-zero vector $x\in L$
is non-degenerate.

Let
\[
\breve{\Bbb P}(\scr V) \subset \Bbb P(\scr V)
\]
be the Zariski open subset obtained by removing $\Bbb P(H)$ from  $\Bbb P(W)$. Then $\breve{\Bbb P}(W)$ is a variety.
 
\subsection{Longitudinals and transversals}

For a variety $X$, let 
\[
\scr V_X 
\] 
denote the trivial vector bundle over $X$ with total space $\scr V\times X$.
Let  $\scr O_{X}(-1)$ be the twisting sheaf considered as a line bundle, let $\tau_X$ denote the tangent bundle of $X$
and let $\tau_X(-1) = \tau_X \otimes \scr O_X(-1)$.
Then in the case of projective space $\Bbb P(\scr V)$, one has 
the  Euler exact sequence of vector bundles
\[
0\to \scr O_{\Bbb P(\scr V)}(-1) \to \scr V_{\Bbb P(\scr V)} \to \tau_{\Bbb P(\scr V)}(-1)\to 0\, .
\]
\begin{lem}
The restriction of Euler sequence to the open subset 
$\breve{\Bbb P}(\scr V) \subset \breve{\Bbb P}(\scr V) $ admits a preferred splitting
\[
\scr V_{\breve{\Bbb P}(\scr V)} =  \scr O_{\breve{\Bbb P}(\scr V)}(-1) \oplus \tau_{\breve{\Bbb P}(\scr V)}(-1) \, .
\] 
\end{lem}
\begin{proof}
Observe that the analytical scalar product induces a scalar product on the bundle $ \scr V_{\Bbb P(\scr V)}$. Moreover, the non-degeneracy
condition implies that
one has a preferred isomorphism of bundles
\[
 \scr V_{\Bbb P(\scr V)} \cong  \scr O_{\breve{\Bbb P}(\scr V)}(-1) \oplus  \scr O_{\breve{\Bbb P}(\scr V)}(-1)^{\perp}\, ,
\]
where the second term on the right denotes the orthogonal complement. 
The latter bundle
is identified with its cokernel of the morphism $ \scr O_{\breve{\Bbb P}(\scr V)}(-1) \to \scr V_{\breve{\Bbb P}(\scr V)} $ and the conclusion follows.
\end{proof}

\begin{defn}
A {\it longitudinal axis} for $\scr V$ consists of a choice of a (closed) point $z \in \breve{\Bbb P}(\scr V)$, or equivalently,
the line $\scr V^\ell \subset \scr V$ given by the fiber at $z$ of the line bundle $\scr O_{\breve{\Bbb P}(\scr V)}(-1)$.
The {\it transversal plane} $\scr V^t$ is the orthogonal complement  to $\scr V^\ell$ in $\scr V$.

A vector $x\in \scr V$ is said to be of transversal (resp.~ longitudinal) type if it lies in the subspace $\scr V^t$ (resp.~ $\scr V^\ell$).
\end{defn}

We note that the transversal/transversal type of a vector depends on the choice of longitudinal axis.

\begin{defn}
A {\it  longitudinal system} consists of a closed point
\[
\frak B := (\scr V_1^\ell,\dots,\scr V_n^\ell) \in \breve{\Bbb P}(\scr V_1) \times \cdots \times  \breve{\Bbb P}(\scr V_n)\, .
\]
\end{defn}

Given $\frak B$ one has an associated super vector space
\begin{equation} \label{eqn:even-odd}
\scr V_I = \scr V_I^e \oplus \scr V_I^o
\end{equation}
in which $\scr V_I^e$  (resp.~ $\scr V_I^e$) is spanned by those basic tensors consisting of an even number (resp.~ odd number) of vectors of transversal type.

\begin{lem}\label{lem:BlochX} A vector $\rho^b \in \scr L^b$
is an  {\it X-state} in the Bloch model  if and only if 
$\rho^b$ lies in the vector subspace
\[
X(\frak B) := \bigoplus_{\emptyset \ne I \subset \overline{n}}  \scr V_I^e \, ,
\]
for some $\frak B$.
\end{lem}

\begin{ex}[$n=2$] \label{ex:n=2} When  $n=2$ and for a given choice of $\frak B$, one has
\[
X(\mathfrak B) =  \scr V^\ell_1\,\, \oplus \,\,\scr V^\ell_2\,\, \oplus\,\,   \scr V_1^t \otimes \scr V_2^t \,\, \oplus \,\,  \scr V_1^\ell \otimes \scr V_2^\ell  \, .
\]
\end{ex}

  \begin{proof}[Proof of Lemma \ref{lem:BlochX}]
A choice of ordered
basis for each qubit allows one to write an L-state
as a linear combination of projection operators of the form
\begin{equation} \label{eqn:projection}
|\phi(1)\cdots \phi(n)\rangle\langle \psi(1)\cdots \psi(n)| = \bigotimes_{i=1}^{n}  |\phi(i) \rangle\langle  \psi(i)|\, ,
\end{equation}
for binary sequences $\phi,\psi\: \overline{n} \to \{0,1\}$. Note that the  $|\phi(i) \rangle\langle  \psi(i)| \in \scr V_i$. 

 An L-state is an
X-state if and only it is a linear combination of such operators constrained by the condition that 
 $\sum_i \phi(i) \equiv \sum_i \psi(i) \mod 2$. Equivalently, if we let $J \subset n$ be the set of indices $i$
 in which $\phi(i) \ne \psi(i)$, then to obtain an X-state, we require
 \[
 \sum_{i\in J} \phi(i) \equiv\sum_{i\in J} \psi(i) \mod 2\, .
\]
But it is easy to see that this condition is equivalent to the condition that the indexing set $J$ having even cardinality.

Moreover, the basis set for the qubit $V_i$ induces a basis of Pauli matrices for the associated inner product space $\scr V_i$
with the span of $\sigma_z$  giving the longitudinal $\scr V_i^\ell$, and the span of $\sigma_x,\sigma_y$ as the transversal subspace $\scr V_i^t$. We may write these operators as linear 
combinations of projections as follows:
\[
\sigma_x = |0 \rangle \langle 1| +  |1 \rangle \langle 0|, \quad \sigma_y = -i|0 \rangle \langle 1| +  i|1 \rangle \langle 0|, \quad
 \sigma_z = |0 \rangle \langle 0| -  |1 \rangle \langle 1|\, .
\]
By writing each  projection of the form $|r\rangle \langle s|$ as a linear combination of the Pauli matrices, one sees that 
the operators $|0\rangle \langle 1|, |1\rangle \langle 0|$ are elements of $\scr V_i^t$. It follows that
when $\phi(i) \ne \psi(i)$, the operator $|\phi(i) \rangle \langle\psi(i)|$ is an element of $\scr V_i^t$ and
the expression on the right of \eqref{eqn:projection} defines an X-state in Bloch form if and only if
$\phi(i) \ne \psi(i)$ for an even number of indices $i\in \overline{n}$, i.e., the indexing set $J$ appearing above has even cardinality. 
Hence, there is a bijection between X-states $\rho \in \scr L$ to and vectors $\rho^b \in \scr L^b$ which lie in $X(\frak B)$ for
a variable longitudinal system $\frak B$.
\end{proof}

\begin{notation} Henceforth, when referring to L-states we will only be making use of the Bloch model. If no confusion arises, we will remove the superscript
from $\scr L^b$. That is,  $\scr L$ now denotes the Bloch model for the space of L-states.
\end{notation}
  
\begin{defn}\label{defn:X-states} 
The {\it variety of  X-states} is
the  $\Bbb C$-subscheme 
\[
\scr X \subset \scr L
\] 
given by the  reduced induced subscheme structure on the Zariski closure of the set of X-states.
\end{defn}

The following justifies calling $\scr X$ a variety.

\begin{lem} The scheme $\scr X$ is a variety with 
 $G$-action. 
\end{lem}

\begin{proof} Observe that $G$ acts continuously
on $\scr L$ and the set of X-states is an invariant subset with respect to the action. 
It follows by continuity that $G$ acts on $\scr X$.

The scheme $\scr X$ is of finite type since 
it is a subscheme of the affine space $\scr L \cong \Bbb A^{2^n-1}$.
It suffices to prove that
$\scr X$ is irreducible.
As $X(\frak B)$ is an affine space, it is irreducible. The
orbit $G\cdot X(\frak B)$ is irreducible as $G$ is connected. Then
the closure of the orbit is irreducible by \cite[ex.~1.1.4]{Hartshorne}.
\end{proof}

\subsection{A vector bundle model for X-states}
We shall describe a vector bundle 
\[
\scr E \to \textstyle{\prod}_{i=1}^n \breve{\Bbb P}(\scr V_i)
\]
 whose total space $\scr E$ is birationally equivalent to $\scr X$
 in which the fiber at  $\frak B \in \prod_{i=1}^n \breve{\Bbb P}(\scr V_i)$ is 
  the vector space $X(\frak B)$ of Lemma \ref{lem:BlochX}.

To avoid clutter, it will be convenient to introduce some additional notation.
We define
\[
 \breve {\Bbb P}_i :=  \breve {\Bbb P}(\scr V_i)\, ,
\]
and for $I \subset \overline{n}$, we define
\[
\breve  {\Bbb P}_I := \prod_{i\in I}\breve {\Bbb P}_i\, .
\]
Let $p_I\: \breve{\Bbb P}_{\overline{n}}  \to \breve{\Bbb P}_I$ be the projection.

Additionally, we set
\begin{itemize}
\item $\scr O_i := \scr O_{ \breve {\Bbb P}_i}$,
 \item $\scr O_i(-1) := \scr O_{ \breve {\Bbb P}_i}(-1)$,
\item $\tau_i(-1) :=  \tau_{ \breve{\Bbb P}_i}(-1)$,
\end{itemize}
for $j = 0,1$.

Then for a pair of disjoint subsets $I, J \subset \overline{n}$, we define
\[
\scr O_I(-1) := \boxtimes_{i\in I}  \scr O_i(-1) \quad \text{ and } \quad 
 \quad 
\scr \tau_J (-1) := \boxtimes_{j\in J}  \tau_j(-1)\, ,
\]
where $\boxtimes$ denotes external tensor product.

Then by construction, the vector bundle over $\breve{\Bbb P}_{\overline{n}} $ defined by
\[
 \scr E_{I,J} := p^\ast_I(\scr O_I(-1)) \boxtimes p_J^\ast(\scr \tau_J (-1)) \, ,
\]
has fiber at $\frak B \in \breve{\Bbb P}_{\overline{n}} $ given by 
$\scr V^\ell_I \otimes \scr V^t_J$.

Let $\frak I_n$ denote the set of ordered pairs $(I,J)$ in which $I,J \subset \overline{n}$ are disjoint subsets
such that $I \cup J$ is non-empty and $J$ has even cardinality. 
Then
the Whitney sum
\[
\scr E:= \bigoplus_{(I,J) \in \frak I_n} \scr E_{I,J}
\]
is a vector bundle over $\breve{\Bbb P}_{\overline{n}} $. 
By Lemma \ref{lem:BlochX}, the fiber of this bundle at $\frak B$ is $X(\frak B)$.
This completes the construction of the desired vector bundle. We now investigate its relation to the variety $\scr X$.

There is a
canonical morphism of schemes
\begin{equation} \label{eqn:e_to_l}
\scr E\to \scr L 
\end{equation}
such that the restriction to the fiber of $\scr E$ at $\frak B$  is the inclusion $X(\frak B)\to \scr L$.

\begin{lem} The scheme theoretic image of the morphism
\eqref{eqn:e_to_l} coincides with the variety $\scr X$.
\end{lem}

\begin{proof} The group  $G$ acts transitively on the base $\breve{\Bbb P}_{\overline{n}} $ of $\scr E$. Consequently,
on the level of closed points, the image of the map \eqref{eqn:e_to_l} coincides with the image
of the  composite map
\[
G\times X(\frak B) @> \subset >> G\times \scr L @>\bold \cdot >> \scr L
\]
given by the inclusion followed by the action. But the latter coincides with the 
orbit $G\cdot X(\frak B)$.
\end{proof}

Let $ \accentset{\circ}{\scr L}\subset \scr L$ denote the 
 open subscheme defined by the condition that the projection of a vector to each one-point correlation function
 has non-zero squared norm, i.e., 
 \[
 \accentset{\circ}{\scr L} := \scr L \times_{\prod_i \scr V_i} ( \textstyle{\prod_{i=1}^n} \breve{ \scr V}_i )\, .
 \]
 Let 
 \[
 \accentset{\circ}{\scr E} = \scr E\times_{\scr L} \accentset{\circ}{\scr L}\, .
\]
 Then  $ \accentset{\circ}{\scr E}$ is the base change of $\scr E$ along the open subscheme $ \accentset{\circ}{\scr L}\subset \scr L$.
In particular, 
the base change
\[
 \accentset{\circ}{\scr X} := \scr X \times_{\scr L} \accentset{\circ}{\scr L}
 \]
 coincides with the scheme theoretic image of the morphism $ \accentset{\circ}{\scr E} \to \accentset{\circ}{\scr L}$ (cf.~\cite[\href{https://stacks.math.columbia.edu/tag/081I}{Tag 081I}]{stacks-project}).
 Furthermore,  since base change preserves closed immersions, we have
 
 \begin{lem} The morphism $ \accentset{\circ}{\scr E}  \to \accentset{\circ}{\scr L}$ is a closed immersion whose
 scheme theoretic image is  $ \accentset{\circ}{\scr X}$. In particular, one has an isomorphism 
 $ \accentset{\circ}{\scr E} \cong  \accentset{\circ}{\scr X} $.
 \end{lem}
 
\begin{rem}  One also has an explicit description of $ \accentset{\circ}{\scr E}$.
Define an embedding  $\overline{n}\subset \frak I_n$ by $i \mapsto (\{i\},\emptyset)$.
Then
\[
\frak I_n  \cong \overline{n} \amalg \frak I'_n\, ,
\]
where $\frak I'_n \subset \frak I_n$ is the set
of pairs $(I,J)$ such that $|I|\ne 1$ or $|J| \ne 0$. Define $ \accentset{\circ}{\scr E} \subset \scr E$ as the fiber product
\[
\bigoplus_{i\in \overline{n}} p_i^\ast( \scr O^\times_i(-1))\,\, \times_{\breve{\Bbb P}_{\overline{n}} } \,\,  \bigoplus_{(I,J) \in \frak I'_n}  \scr E_{I,J}
\]
where $\scr O^\times_i(-1)$ is the $\Bbb G_m$-bundle obtained from $\scr O_i(-1)$ by removing the zero section. The fiber
of $ \accentset{\circ}{\scr E}$ at $\frak B$ is identified with the Zariski open subset $\accentset{\circ}X(\frak B) \subset X(\frak B)$ consisting of those
vectors having non-trivial one-point correlation functions.
\end{rem}

\begin{cor} \label{cor:E=X} The morphism of varieties
\[
 \scr E\to \scr X
\]
 is a birational equivalence. 
 \end{cor}

 \begin{proof} In the commutative square
 \[
 \xymatrix{
 { \accentset{\circ}{\scr E}} \ar[r]^{\cong} \ar[d] &  {{\accentset{\circ}{\scr X}}} \ar[d]\\
  \scr E\ar[r] & \scr X
 }
 \]
 the vertical maps are open inclusions and the top horizontal map is an isomorphism.
 The conclusion follows.
\end{proof}

\subsection{The dimension of $\scr X$}

\begin{lem}   The vector space $X(\frak B)$ has dimension
$2^{2n-1}-1$.
\end{lem}

\begin{proof}  Recall that $X(\frak B)$ is spanned by basic tensors 
consisting of an even number of vectors of transversal type. Describing such a tensor involves  choosing  $I,J \subset \overline{n}$
in which $I$ and $J$ are disjoint, $I \cup J$ is non-empty, and $|J|$ is even. 
 
It follows that $X(\frak B)$ is an vector space of dimension
\[
 \sum_{\alpha} \tbinom{n}{\alpha} 2^{t}\, ,
\] 
where $\alpha$ is indexed is over all pairs of non-negative integers $(s,t)$ satisfying
\begin{itemize}
\item $1\le s+t \le n$,
\item $t \equiv 0 \mod 2$.
\end{itemize} Here, 
$
\binom{n}{\alpha} 
$
is shorthand for the  multinomial coefficient $\binom{n}{s,t,n-s-t}$, $s$ corresponds to  $|I|$, $t$ corresponds to $J$, and 
 $2^t$ appears in the sum since $\dim(\scr V^t_i) = 2$ for $i\in \overline{n}$.

We deduce the result by means of the related expression
\[
 \sum_{\beta} \tbinom{n}{\beta} 2^{t}\, ,
\]
where now $\beta$ ranges over all pairs of non-negative integers $(s,t)$ such that $0 \le s + t \le n$.
By the multinomial theorem, this sum counts the number of distinct binary strings of length $2n$, where the 
number of such  strings is $2^{2n}$. Half of them have even sum which gives 
$2^{2n-1}$. If we discard  the case $s+t = 0$, we obtain 
$2^{2n-1}-1$.
\end{proof}

\begin{cor} \label{cor:dim} The variety $\scr X$ has dimension $2^{2n-1}+ 2n - 1$.
\end{cor}

\begin{proof} By Corollary \ref{cor:E=X}, the varieties $\scr X$ and ${\scr E}$ have the same dimension. 
One then uses the fact that the dimension of $\scr E$ is the sum of $\dim(X(\frak B) = 2^{2n-1} -1$ and $\dim (\prod_{i=1}^n \breve{\Bbb P}(\scr V_i)) = 2n$. 
\end{proof}


\section{A rational section \label{sec:rel-sec}}
For a given choice of longitudinal system $\frak B$,  we will show that the vector space $X(\frak B)$ is a rational section for the variety $\scr X$ with
respect to the action of $G$. We will also identify the normalizer and Weyl group of this rational section.

For a  qubit $V$, and a choice of longitudinal $\scr V^\ell \subset \scr V$, the complementary subspace of transversals $\scr V^t$
is an subspace of dimension two and is equipped with the induced analytical scalar product. In particular, the orthogonal group $\Or(\scr V^t)$ of isometries
is a subgroup of $\SO(\scr V)$.

\begin{prop} \label{prop:rational-section} The affine space $X(\frak B)$ is a rational section for the action of the group $G$ on $\scr X$. Furthermore,
the normalizer of this action is given by the product
\[
N := \prod_{i=1}^n \Or(\scr V_i^t) 
\]
and the centralizer is given by $\pm I$, where $I\: X(\frak B) \to X(\frak B)$ is the identity.
\end{prop}

\begin{proof}  Let $\accentset{\circ}X(\frak B) \subset X(\frak B) $ denote the fiber 
at $\frak B$ with respect to the projection 
of $\accentset{\circ}{\scr E} \to \breve{\Bbb P}_{\overline{n}} $.
Then there are nested inclusions
\[
G \cdot\accentset{\circ}X(\frak B)  \subset G \cdot X(\frak B) \subset {\scr X}
\]
in which $G \cdot \accentset{\circ}X(\frak B) $ coincides with $\accentset{\circ}{\scr X} $ 
since $G$ acts transitively on $\breve{\Bbb P}_{\overline{n}} $. These inclusions become equalities 
after taking closures, so the closure of the orbit $G \cdot \accentset{\circ}{X}(\frak B)$ is $\scr X$.

Since $\accentset{\circ}X(\frak B)$ is the fiber at $\frak B$ of $\accentset{\circ}{\scr E} \to \breve{\Bbb P}_{\overline{n}} $
and $\accentset{\circ}{\scr E}  \to \accentset{\circ}{\scr X}$ is an isomorphism, 
it follows that $\accentset{\circ}X(\frak B)  \cap \accentset{\circ}X(g\frak B) \ne \emptyset$ 
if and only if $g\frak B = \frak B$, i.e., if and only if $g = (g_1,\dots,g_n)$ stabilizes the longitudinal system $\frak B$.
But this means that each $g_i$ stabilizes $\scr V_i^\ell$ which is equivalent to the assertion
that $g_i \in \Or(\scr V^t_i)$ for all $i$. We infer that $\accentset{\circ}{X}(\frak B) \cap \accentset{\circ}{X}(g\frak B) \ne \emptyset$
if and only if $g\in N$.
Consequently, $X(\frak B)$ is a rational section with normalizer $N$.

The centralizer consists of those elements $g\in N$ which fix ${X}(\frak B)$. In particular,
$g$ is in the centralizer if it fixes every non-zero one-point correlation function. This can only be the case if $g = \pm I$.
 \end{proof}
 
 Let $W = N/\{\pm I\}$ denote the Weyl group. Then by Proposition \ref{prop:sec-rel} we infer
 
 \begin{cor} \label{cor:section} The inclusion $X(\frak B) \to  \scr X$ induces an isomorphism 
 \[
 \Bbb C(\scr X)^G \cong \Bbb C(X(\frak B))^W \, .
 \]
 \end{cor}

\section{Overview} \label{sec:overview}

The goal of this section is to outline the proof of Theorem \ref{bigthm:main}.
By Corollaries \ref{cor:dim} and  \ref{cor:section}, it will suffice to prove that $X(\frak B)/\!\!/W$ is rational. 

In broad strokes, the idea is to apply  Proposition \ref{prop:no-name}
to a certain $W$-linearized vector bundle 
\begin{equation} \label{eqn:W-lin}
X(\frak B) \to X_T(\frak B)
\end{equation}
of dimension $2^{2n-1}-5n+3$
to deduce that  $X(\frak B)/\!\!/W$ birationally equivalent to
the total space of  a vector bundle over the algebraic quortient $X_T(\frak B)/\!\!/W$.
The bulk of our effort is then to show that
the $X_T(\frak B)/\!\!/W$
is rational of dimension $4n-4$.

Here is an outline of  the main steps in slightly more detail.

\begin{enumerate}[Step (i):]
\item Let $K_n$ denote the complete graph on the set vertices $\overline{n}$.  The edges of $K_n$ 
are designated by their endpoints: if $i$ and $j$
denote distinct vertices, then $ij$ denotes the associated edge.

Choose a spanning
tree
\[
T\subset K_n\, .
\]
\item Define
\[
X^{0}(\frak B) := \bigoplus_{i\in T_0}  \scr V_i^\ell \, , \quad X^{1}(\frak B)  := \bigoplus_{ij\in T_1}  \scr V_{ij}^t \, ,
\]
and
\[
X_T(\frak B)  := X^{0}(\frak B)  \oplus X^{1}(\frak B) \, .
\]
Then $X_T(\frak B) \subset X(\frak B)$ is a split summand and projection to this summand gives a $W$-equivariant surjective homomorphism of vector spaces
\[
X(\frak B) \to X_T(\frak B)
\]
with kernel having dimension $2^{2n-1}-5n+3$.
In particular, the displayed map  $W$-linearized vector bundle of dimension
$2^{2n-1}-5n+3$, thereby producing \eqref{eqn:W-lin}.

\item 
One checks that the Weyl group $W$ acts almost freely on  $X_T(\frak B)$.
 Applying Proposition \ref{prop:no-name}
one deduces that $X(\frak B)/\!\!/W$ is birationally equivalent
to the total space of a vector bundle over $X_T(\frak B)/\!\!/W$.
We are now reduced to the problem of showing that $X_T(\frak B)/\!\!/W$ is rational.

\item  One shows that $X_T(\frak B)/\!\!/W$ is rational in multiple stages:
The connected component $W'$ of $W$ is $\SO_2(\Bbb C)^{\times n}$, and there is a short exact sequence 
\begin{equation} \label{eqn:wprime}
1\to W' \to W\to (\Bbb Z/2)^{\times n} \to 1
\end{equation}
of group schemes over $\Bbb C$. Making use of Lemma \ref{lem:double-quotient1} and Lemma  \ref{lem:double-quotient2},
we first construct a rational quotient with respect to the action of $W'$ on $X_T(\frak B)$  and then subsequently
a rational quotient with respect to $(\Bbb Z/2)^{\times n}$ acting on $X_T(\frak B)/\!\! /W'$. 
The model we construct for the rational quotient 
 $X_T(\frak B)/\!\!/W$ is then seen to be rational of dimension $4n-4$.
 \end{enumerate}

\begin{rem}\label{rem:wprime}  In the short exact sequence \eqref{eqn:wprime}, we implicitly  used the fact that over $\Bbb C$,  there is an isomorphism of
group schemes 
\[
\mu_2 \cong \Bbb Z/2\, ,
\]
where
$\mu_2$, the square roots of unity,  is represented by $\spec(\Bbb C[t]/(t^2-1))$ and $\Bbb Z/2$ is represented by  $\spec(\Bbb C\times \Bbb  C)$.
\end{rem}

\section{An auxiliary problem}  \label{sec:aux}
A crucial ingredient in step (v) of  the previous section is to  identify
a rational quotient  for the left action of $\SO_2(\Bbb C)$ on the vector space $M_2(\Bbb C)$  of  $2\times 2$ matrices over the complex numbers.

The group $\SO_2(\Bbb C)\times \SO_2(\Bbb C)$ acts 
on $M_2(\Bbb C)$ by
\[
(g,h)\cdot M = gMh^{-1}\, .
\]
To distinguish the factors, we write $\SO_2(\Bbb C)_\ell$ and $\SO_2(\Bbb C)_r$ for the 
left and right factors. Let $D \subset \SO_2(\Bbb C)_\ell \times \SO_2(\Bbb C)_r$  be the subgroup
isomorphic to $\Bbb Z/2$ given by $\{(I,I),(-I,-I)\}$, where $I$ is the identity matrix. 
Then $D$ acts trivially on $M_2(\Bbb C)$ so the action factors through the quotient group
\[
 (\SO_2(\Bbb C)_\ell \times \SO_2(\Bbb C)_r)/D\, .
\]
Let $\PSO_2(\Bbb C) = \SO_2(\Bbb C)/(\pm I)$.

\begin{lem} \label{lem:untwist} There is a preferred isomorphism
\[
 (\SO_2(\Bbb C)_\ell \times \SO_2(\Bbb C)_r)/D \cong  \SO_2(\Bbb C)_\ell \times \PSO_2(\Bbb C)_r 
 \]
 as well as a preferred isomorphism $\PSO_2(\Bbb C)_r \cong \SO_2(\Bbb C)$.
 \end{lem}
 
 \begin{proof} An isomorphism $\SO_2(\Bbb C)_\ell \times \PSO_2(\Bbb C)  \to  (\SO_2(\Bbb C)_\ell \times \SO_2(\Bbb C)_r)/D$ is defined by
 $(A,\pm B) \mapsto (A,B)D = \pm (A,B)$. The  isomorphism $\PSO_2(\Bbb C)\to \SO_2(\Bbb C)$ is defined by $\pm A\mapsto A^2$.
 \end{proof}

\subsection{Digression: the multiplicative group}
 Let $\Bbb G_m = \Bbb C^{\times}$ be the multiplicative group of the complex numbers.
Recall that $\SO_2(\Bbb C)$ is the complex Lie group given by the set complex matrices
\[
A = \begin{pmatrix}
a & b \\
-b & a
\end{pmatrix}
\]
in which $a^2 + b^2 = 1$.
Define a homomorphism
\[
h\: \SO_2(\Bbb C) \to \Bbb G_m
\]
by sending the above matrix $A$ to $a+bi$.

\begin{lem}  \label{lem:mult} The homomorphism $h$ is an isomorphism.
\end{lem}

\begin{proof} The inverse to $h$ maps $\lambda \in \Bbb G_m$ to the 
matrix
\[
\begin{pmatrix}
\tfrac{\lambda + \lambda^{-1}}{2} &  \tfrac{\lambda - \lambda^{-1}}{2i}  \\
-\tfrac{\lambda - \lambda^{-1}}{2i} & \tfrac{\lambda + \lambda^{-1}}{2} 
\end{pmatrix} \qedhere
\]
\end{proof}
 
\subsection{Solution to the auxiliary problem}
We now return to the problem of providing rational quotient for the action of 
$\SO_2(\Bbb C)_\ell$ on $M_2(\Bbb C)$
as a variety with action of $\PSO_2(\Bbb C)_r \cong \SO_2(\Bbb C)$.

The operation $M \mapsto M^tM$ defines a homomorphism  
\[
M_2(\Bbb C) \to \symm_2(\Bbb C)\, ,
\] where the target denotes the
vector space of symmetric $(2\times 2)$-matrices. If $M \in M_2(\Bbb C)$, then one may uniquely write
\[
M^tM= 2tI + A
\]
where $A$ is traceless and symmetric, where $2t$ is the trace of $M^tM$. Hence, $A$ is of the form
\[
\begin{pmatrix}
a & b \\
b & -a
\end{pmatrix}
\]
for suitable $a,b \in \Bbb C$. Then these quantities define $(\SO_2(\Bbb C)_\ell)$-invariant functions 
\[
\delta,t,a,b
\]
in which $\delta(M) = \det M, t(M) = \frac{1}{2}\tr(MM^t)$ and $a$ and $b$ are as above.
These functions satisfy the polynomial relation
\[
\delta^2 = t^2 - a^2 - b^2\, .
\]
Let $f = \delta^2 - t^2 + a^2+ b^2$. Then the principal ideal $(f) \subset k[\delta,t,a,b]$ is prime.

 Let \[
 R = (\Bbb C[\delta,t,a,b]/(f))_\delta
 \] 
 be the localization of $\Bbb C[\delta,t,a,b]/(f)$ obtained by
 inverting $\delta$. Let
 \[
 U :=\spec(R)\, .
 \] 
Then $U$ is a variety over $\Bbb C$ equipped with the action of $ \PSO_2(\Bbb C)_r$.

\begin{prop}\label{prop:torsor} The variety $U$
is a rational quotient for the action of  $\SO_2(\Bbb C)_\ell$ on $M_2(\Bbb C)$.

Furthermore, $\SO_2(\Bbb C) \cong \PSO_2(\Bbb C)_r$ acts trivially on the generators
$t$ and $\delta$ and by left multiplication on $(a,b)$ when the latter is considered as a column vector.
\end{prop}

 \begin{proof} Let 
$Y \subset M_2(\Bbb C)
 $ be the Zariski open subset consisting of the matrices $M$ satisfying the constraint
 \[
 \delta(M)\ne 0\, .
 \]
 For the first part, it will be enough to prove that the map $Y \to U$ defined by 
 \[
 M \mapsto (\delta,t,a,b)
 \]  is an $\SO_2(\Bbb C)_\ell$-torsor.
 
We may think of the action as follows: a matrix $M\in Y$ is a pair of column vectors
 $u = (u_1,u_2),  v = (v_1,v_2)$ and $g\in \SO_2(\Bbb C)$ acts by acting on these vectors
 individually: $g\cdot (u,v) = (gu,gv)$.  We define
 \[
 u_\pm = u_1 \pm iu_2\, , \quad v_\pm = v_1 \pm iv_2\, .
 \]
 This gives a change of basis to a pair of column vectors 
 \[
 u_\pm := (u_-,u_+), v_\pm :=  (v_-,v_+)\, .
 \]
If we identify $\SO_2(\Bbb C)$ with $\Bbb G_m$ (using Lemma \ref{lem:mult}), then the action in this new basis becomes
\[
\lambda\cdot (u_-,u_+) = (\lambda u_-,\lambda^{-1} u_+)\, , \quad \lambda\cdot (v_-,v_+)  = (\lambda v_-,\lambda^{-1} v_+)\, .
\]

Then $\delta \ne 0$ is the same as asking that
\[
u_-v_+ - v_-u_+ \ne 0\, .
\]
We may then write $Y = Y_- \cup Y_+$, where $Y_\pm$ is defined by the condition $u_\pm v_\mp \ne 0$.
Note that $\Bbb G_m$ acts freely on $Y_\pm$.
Given $M \in Y_-$, say
\[
M:= \begin{pmatrix}
u_- & v_- \\
u_+ & v_+
\end{pmatrix}\, ,
\]
we may act by $u_-^{-1} \in \Bbb G_m$ to obtain
\[
u_-^{-1}M := \begin{pmatrix}
1 & u_{-}^{-1} v_- \\
u_-u_+ & u_-v_+
\end{pmatrix}
\]
and the latter is the unique point in the orbit of $M$ with the property that the $(1,1)$-entry is 1. 
In the above, the $(2,2)$-entry is non-zero, whereas the other entries have no constraint.
It follows that $Y_- \to U_-$ is a trivializable $\Bbb G_m$-torsor with  $U_- \cong \Bbb G_m \times \Bbb A^2$. Similarly, $Y_+ \to U_+$ is a trivializable $\Bbb G_m$-torsor and therefore $Y\to U$ is
a $\Bbb G_m$-torsor.

As to the last part, $\PSO_2(\Bbb C)_r$ acts by conjugation on the matrix $M^tM = 2tI + A$, but conjugation on $2tI$ is trivial. Hence the action on $t$ is trivial.
With respect to the isomorphism $\PSO_2(\Bbb C)_r\cong \SO_2(\Bbb C)$, a straightforward calculation that we omit shows that the action of $\SO_2(\Bbb C)$
is by multiplication on the column vector associated with $(a,b)$. Lastly, the action of $\PSO_2(\Bbb C)_r$ on $\delta = \det(M)$ is trivial since the determinant
is invariant with respect to multiplication by elements of $\SO_2(\Bbb C)_\ell\cong \SO_2(\Bbb C)_r$.
\end{proof}

\section{The quotient by $W'$} \label{sec:quotient-by-wprime}
Recall from \S\ref{sec:rel-sec} that 
\[
N = \prod_{i=1}^n \Or(\scr V_i^t)
\]  is the normalizer of the section $X(\frak B)$ and $W = N/\{\pm I\}$ is the  Weyl group.
Consider the maximal connected subgroup 
\[
N' \subset N
\]
given by
\[
N' =  \prod_{i=1}^n \SO(\scr V_i^t) \, .
\]
Then one has  short exact sequences
\[
1\to N' \to N \to (\Bbb Z/2)^{\times n} \to 1\, , \qquad 1\to W' \to W \to (\Bbb Z/2)^{\times n} \to 1\, ,
\]
where $W' = N'/\{\pm I\}$  (cf.~\eqref{eqn:wprime} and Remark \ref{rem:wprime}). 

The goal of this section is to show that
$X_T(\frak B)/\!\!/W'$ is rational by constructing a rational quotient
for the action of $W'$ on 
\[
X_T(\frak B) = X^0(\frak B) \oplus X^1(\frak B)
\]
(cf. ~\S\ref{sec:overview}).
Note that $W'$ acts trivially on $X^0(\frak B)$. 
Consequently,  the algebraic quotient for $W'$ acting on $X_T(\frak B)$ is
\[
X^0(\frak B) \times (X^1(\frak B)/\!\!/W' )\, .
\]
Hence, it will suffice to construct a rational quotient for the action of $W'$ on $X^1(\frak B)$ and show that it is rational.

We will first construct a rational quotient  with respect to the subgroup of $W'$ given by
\[
W'_{n-1} := \prod_{i=1}^{n-1} \SO(\scr V_i^t)\, .
\]
The proof of the following result is essentially
the same as the proof of Lemma \ref{lem:untwist}; we omit the details.

\begin{lem} \label{lem:prod-iso} There is a preferred isomorphism
\[
W' \cong W'_{n-1}  \times  \PSO(\scr V_n^t)\, ,
\]
as well as a preferred isomorphism $ \PSO(\scr V_n^t) \cong  \SO(\scr V_n^t)$.
\end{lem}

\begin{rem} With respect to the isomorphism of Lemma \ref{lem:prod-iso}, it is straightforward to check that the action of $W'$ on 
the summand 
\[
\scr V_i^t \otimes \scr V_j^t \cong \hom(\scr V_j^t, \scr V_i^t )
\]
of $X^1(\frak B)$ factors through the projection $W' \to \SO(\scr V_i^t) \times \SO(\scr V_j^t)$, provided
$i,j < n$.
\end{rem}

\subsection{The quotient of  $X^1(\frak B)$ by $W_{n-1}'$ } 
For convenience, we will assume:

\begin{hypo} The spanning tree $T$ is a star graph with center vertex $n$. In particular, every edge of $T$ is of the form $jn$ for some $j \in \{1,\dots,n-1\}$.
\end{hypo}

\begin{rem} Although this hypothesis is not necessary, it does simplify the proof.
\end{rem} 

In what follows we choose an orthonormal basis for each $\scr V_i^t$. Then we have an identification 
\[
\scr V^t_j\otimes \scr V_n^t \cong M_2(\Bbb C^2)
\]
for $j \in \{1,\dots,n-1\}$. 
Let
\[
W'_{n-1} =\prod_{i=1}^{n-1} \SO(\scr V_i^t)\, .
\]
Let $\Bbb C[\delta_\bullet,t_\bullet, a_\bullet,b_\bullet]$ denote the
polynomial ring generated by $\delta_i,t_i,a_i,b_i$ for $i = 1,\dots n-1$.
Let $(f_\bullet)$ be the ideal generated by $\{f_1,\dots,f_{n-1}\}$, and let
\[
(\Bbb C[\delta_\bullet,t_\bullet, a_\bullet,b_\bullet]/(f_\bullet))_{\delta_\bullet}
\]
be the localization of $\Bbb C[\delta_\bullet,t_\bullet, a_\bullet,b_\bullet]$ obtained by inverting
the product $\delta_1\cdots\delta_{n-1}$.

\begin{lem} \label{lem:quotient-w-n-1} A rational quotient  for $W'_{n-1}$ acting on $X^1(\frak B)$ is given by 
\[
\spec(\Bbb C[\delta_\bullet,t_\bullet, a_\bullet,b_\bullet]/(f_\bullet))_{\delta_\bullet}
\]
with induced action of 
\[
W'/W'_{n-1} = \PSO(\scr V^t_n) \cong \SO_2(\Bbb C)
\]
as in Proposition \ref{prop:torsor}.
\end{lem}

\begin{proof} 
By an iterated application of Lemma \ref{lem:double-quotient1}
we have an isomorphism of algebraic quotients
\[
X^1(\frak B)/\!\!/W'_{n-1} =  \prod_{i=1}^{n-1} M_2(\Bbb C)/\!\!/\SO_2(\Bbb C)_\ell \, .
 \]
 Consequently, the result follows from Proposition \ref{prop:torsor}.
\end{proof}

\subsection{The quotient of  $X^1(\frak B)$ by $W'$} 
The isomorphism
\[
W' \cong W'_{n-1} \times \SO(\scr V_n^t) 
\]
and Lemma \ref{lem:quotient-w-n-1} show that
a rational quotient for the action of $W'$ on  $X^1(\frak B)$  is given by a rational quotient
for the action of $\SO(\scr V_n^t)$ on $X^1(\frak B)/\!\!/W'_{n-1}$. 
Our next step is to construct a rational quotient for the latter action using 
the rational quotient provided by Lemma \ref{lem:quotient-w-n-1}.

Let
 $\{e_1,e_2\}$ be an  orthonormal basis for $\scr V_n^t$.
Then 
\[
 \SO(\scr V_n^t)  \cong \SO_2(\Bbb C) \, .
 \]
So we have identifications $\Bbb C^2 \cong  \scr V_i^t$ for $i = 1,\dots, n$.

Consider the alternative basis $\{e_\pm\}$ for $\Bbb C^2$ given by
\[
e_\pm = e_1 \pm ie_2\, .
\]
If 
\[
w_i := (a_i,b_i) \in \Bbb C^2\, ,
\]
we let $w_i^\pm$ denote the components of $w$ in the new basis,
i.e., $w_i = w_i^-e_- + w_i^+ e_+$. Consequently, in the alternative basis,
 $w_i$  has coordinates $(w_i^-,w_i^+)$.

Using the isomorphism $\Bbb G_m \cong \SO_2(\Bbb C)$ of Lemma \ref{lem:mult}, 
one calculates that the action of $\Bbb G_m$ on $\Bbb C^2$ in the alternative basis is given by
\[
\lambda\cdot (w_i^-,w_i^+) = (\lambda w_i^-,\lambda^{-1}w_i^+)\, .
\]
Note that $f_i = \delta_i^2 - t_i^2 + a^2 + b^2$ as expressed  in the alternative basis is
the $\Bbb G_m$-invariant function
\[
\hat f_i := \delta_i^2 - t_i^2 + 4w_i^-w_i^+\, ,
\]
and as expressed in the new variables, we obtain the following alternative 
to Lemma \ref{lem:quotient-w-n-1}.

\begin{add} A rational quotient for  $W'_{n-1}$ acting on $X^1(\frak B)$ is given by
\[
 \spec(\Bbb C[\delta_\bullet ,t_\bullet, w_\bullet^-,w_\bullet^+]/(\hat f_\bullet)_{\delta_\bullet}) \, ,
\]
where $\Bbb G_m$ acts trivially on the generators $\delta_i,t_i$, 
and on the variables $w_i^\pm$ by $\lambda\cdot w_i^- = \lambda w_i^-$,
 and  $\lambda\cdot w_i^+ = \lambda^{-1} w_i^+$. 
 \end{add}
 
 The advantage of the latter description is that it allows
 an explicit description of a rational quotient for the action of $W'/W'_{n-1}$ on $X^1(\frak B)$.

\begin{defn} Let $\Omega$ be the set of monomials generated by $\{w_i^\pm\}_{i\le n-1}$. Then $\Omega$ is a monoid under multiplication. 
The {\it signature} is the unique monoid homomorphism $\Omega \to (\Bbb Z,+)$
which maps $w_i^\pm$ to $\pm 1$.
\end{defn}

Clearly, a monomial in the generators $\{w_i^\pm\}_{i\le n-1}$ is  $\Bbb G_m$-invariant if and only its signature is zero. 
Let  $\Omega^0$ to be the set
of all monomials in the generators $w_i^\pm$ that have signature zero. Then 
\[
u_{jk} := 4w_j^+w_k^-  \quad 1\le j,k \le n-1
\]
are a set of generators for $\Omega^0$.

\begin{defn}
A {\it loop} $\gamma$ of length $r$ in the complete graph $K_{n-1}$ is a finite sequence of contiguous, directed edges of the form
\[
\epsilon_1 = (j_1,j_2), \epsilon_2 = (j_2,j_3),\cdots ,\epsilon_r =(j_r,j_1)
\]
The set of vertices appearing in $\gamma$ is denoted by $v(\gamma)$, and the set of directed edges appearing in $\gamma$ is denoted by
$e(\gamma)$. We write $\ell(\gamma) = r$ for the length of $\gamma$.
\end{defn}

For every 
loop $\gamma$, we  define 
\[
u_\gamma := \prod_{\alpha \in e(\gamma)} u_{\alpha}\, .
\] 
Then  we have the relation
\[
u_\gamma   =  \prod_{j\in v(\gamma)}  (t_{j}^2 - \delta_{j}^2) \,  .
\]

Since  any loop can be rewritten as a concatenation of loops of length at most three, we have the following.

\begin{lem} A complete set of relations is given by
\[
u_{kk} = t_k^2 -  \delta_k^2 , \qquad u_\gamma   =  \prod_{j\in v(\gamma)}  (t_{j}^2 - \delta_{j}^2)\, , 
\]
as $\gamma$ ranges over all loops of length $\le 3$ and $1 \le k \le n-1$.
\end{lem}

The relations can further be simplified as follows: if $\alpha = (i,j)$ is a directed edge of $K_{n-1}\subset K_n$
then we write 
\[
\tilde u_{ij} := \frac{u_{ij}}{\delta_i^2 - t_i^2}\, , \qquad i \ne j\, .
\] 
Then we have 
$\tilde u_\gamma = 1$ for all loops $\gamma$. In particular $\tilde u_{ij} ^{-1} = \tilde u_{ji}$.
Furthermore,  for $j,k \ne 1$, the relation
\[
\tilde u_{jk} = \frac{\tilde u_{1k}}{\tilde u_{1j}}
\]
shows that the relations are determined by the loops of length $\le 3$ that are based at the basepoint vertex $1$:

\begin{cor} A complete set of relations is given by
\[
\tilde u_\gamma = 1\, .
\]
where $\gamma$ ranges over the set of loops of length $\le 3$ that contain the basepoint vertex $1$.
\end{cor}

Set
\[
\tilde u_j := \tilde u_{1j} \quad j = 2,\dots n-1\, ,
\]
and define
\[
\Delta = \prod_{i=1}^{n-1} \delta_i(\delta^2_i-t_i^2)\, .
\]
Then there are no relations between the variables $\tilde u_j$.
 Consequently,

\begin{prop} \label{prop:wprimeX1} A rational quotient for the action of $W'$ on $X^1(\frak B)$ is given by
\[
\spec(\Bbb C[t_\bullet,\delta_\bullet,\tilde u_\bullet]_{\Delta})
\]
In particular, the field $\Bbb C(X^1(\frak B))^{W'}$ is pure of transcendence degree $3n-4$.
\end{prop}

We conclude this section by determining a rational quotient for the action of $W'$ on $X_T(\frak B)$.
Let
\[
e_\bullet = \{e_1, \dots, e_n\}
\]
denote an orthonormal basis for the vector space $X^0(\frak B)$, where  $e_j$ spans $\scr V_j^\ell$.
We write
\[
\alpha_j\: \scr V^\ell_j \to \Bbb C\, , \qquad j =1,\dots, n
\]
for the invariant function given by $\alpha_j(v) = \langle v,e_j\rangle$.

Recall that $W'$ acts trivially on $X^0(\frak B)$. Hence,

\begin{cor} \label{cor:wprimeX1} 
A rational quotient for the action of $W'$ on $X_T(\frak B)$ is given by
\[
\spec(\Bbb C[t_\bullet,\delta_\bullet,\tilde u_\bullet,\alpha_\bullet]_{\Delta})\, .
\]
In particular, the field $\Bbb C(X_T(\frak B))^{W'}$ is purely transcendental over $\Bbb C$ of degree $4n-4$.
\end{cor}

\section{The quotient by $W/W'$}  \label{sec:w-wprime}
In order to understand how the group
\[
W/W' \cong (\Bbb Z/2)^{\times n}
\]
 acts on the rational quotient  provided by Corollary \ref{cor:wprimeX1},
we introduce yet another set of generators $s_j,v_j$ that are defined by the equations
\[
s_j := \frac{1}{2}(u_{1j} + u_{j1}) \quad \text{and} \quad v_j := \frac{1}{2i}(u_{1j} - u_{j1})\, , \quad j = 2,\dots n-1\, .
\]
These generators are subject to the relations
\[
\rho_j := s_j^2 + v_j^2 - (t_j^2-\delta_j^2)(t_1^2-\delta_1^2) = 0\,.
\]
Then $X_T(\frak{B})/\!\!/W'$ is birationally equivalent to $\spec$ of
\[
(\Bbb C[t_\bullet,\delta_\bullet, s_\bullet,v_\bullet,\alpha_\bullet]/(\rho_\bullet))_{\Delta} \, .
\]
If we identify $W/W' = (\Bbb Z/2)^{\times n}$, then 
\begin{itemize}
\item all factors of $W/W'$ act trivially on the variables $s_j$ and $t_j$;
\item only the last factor of $W/W'$  acts non-trivially on the variables $v_j$;
\item only the first and $j$-th factors  of $W/W'$  act non-trivially on the variables $\delta_j$;
\item only the $j$-th factor of $W/W'$ acts non-trivially on $\alpha_j$.
\end{itemize}

Define yet another generating set by
\begin{alignat*}{3}
\tilde s_j &= s_j\, , \quad &\tilde v_j = \alpha_nv_j \, , \quad  &j = 2,\dots, n-1\, ,  \\
\tilde t_k &= t_k\, ,  \quad &\tilde \delta_k = \alpha_k\alpha_1\delta_k\, ,  \quad &k = 1,\dots n-1\, ,\\
\eta_\ell &= \alpha_\ell^2\, , &&\ell =1 \dots n\, .
\end{alignat*}
Then $W/W'$ acts trivially on the new generating set. Moreover, the relation defined by $\rho_j$ becomes:
\[
\tilde s_j^2 + \frac{\tilde v_j^2}{\eta_n}  = (\tilde t_j^2 - \frac{\tilde\delta_j^2}{\eta_1\eta_j})(\tilde t_1^2 - \frac{\tilde \delta_1^2}{\eta_1^2})\, .
\]
Solving for $\eta_j$ when $j\ne 1,n$ and simplifying, we obtain the rational function
\begin{equation} \label{eqn:rational-fcn}
\eta_j = \frac{\tilde \delta^2\eta_n(\tilde t_1^2\eta_1^2 - \tilde \delta_1^2)}{\eta_1\eta_n\tilde t_j^2(\tilde t_1^2\eta_1^2  -\tilde \delta_1^2) - \eta_1^3(\tilde s_j^2\eta_n + \tilde v_j^2) }
\end{equation}
It follows that the generators $\eta_j$ for $j \ne 1,n$ can be discarded.
Let $\beta_j$ denote the denominator of \eqref{eqn:rational-fcn} and define
\[
\tilde \Delta :=
 \eta_1\eta_n (\eta_1^2\tilde t_1^2 - \tilde \delta_1^2)(\prod_{k=1}^{n-1} \tilde \delta_k)(\prod_{j=2}^{n-1}\beta_j(\eta_n \tilde s_j^2 + \tilde v_j^2) ).
\]
Then 
\begin{cor}  \label{cor:rat-quo-w-XT} A rational quotient for the action of $W$ on $X_T(\frak B)$ is given by
\[
\spec(\Bbb C[\tilde t_\bullet, \tilde \delta_\bullet,\tilde s_\bullet,\tilde v_\bullet,\eta_1,\eta_n]_{\tilde \Delta})\, .
\]
In particular, the algebraic quotient $X_T(\frak B)/\!\!/W$ is rational and 
the function field $\Bbb C(X_T(\frak B))^W$ 
is purely transcendental over $\Bbb C$ of degree $4n-4$.
\end{cor}

\section{The proof of Theorem \ref{bigthm:main}}  \label{sec:main}

We are now in a position to complete the proof of the main result.
Let $X(\frak B)$ and $W$ be as above. By Proposition \ref{prop:rational-section}, $X(\frak B)$ is a rational section for the action of $G$ on 
$\scr X$. By Proposition \ref{prop:sec-rel}, the restriction homomorphism
\[
\Bbb C(\scr X)^G \to \Bbb C(X(\frak B))^W
\]
is an isomorphism. Consequently, it suffices to show that $\Bbb C(X(\frak B))^W$
is purely transcendental over $\Bbb C$ of degree $2^{2n-1} - n-1$.

By Corollary \ref{cor:no-name} applied to the $G$-equivariant short exact sequence of vector spaces
\[
0 \to F\to X(\frak B) \to X_T(\frak B)\to 0
\]
we obtain  an isomorphism of fields
\begin{equation}\label{eqn:last}
\Bbb C(X(\frak B))^W \cong  \Bbb C(F) \otimes \Bbb C(X_T(\frak B))^W  \, .
\end{equation}
Furthermore, the vector space $F$ has dimension  $2^{2n-1} - 5n+3$.

Choosing $T$ to be the spanning tree of $K_n$ given by the star graph
on $n$-vertices with center the vertex $n$, we apply
Corollary \ref{cor:rat-quo-w-XT} to deduce that
$\Bbb C(X_T(\frak B))^W$ is purely
transcendental over $\Bbb C$ of degree $4n-4$.  

By \eqref{eqn:last} we infer that $\Bbb C(X(\frak B))^W$ is purely transcendental over $\Bbb C$
of degree $(2^{2n-1} - 5n+3) + (4n-4) = 2^{2n-1} - n-1$. \qed

\appendix 
\section{Descent revisited}  \label{sec:descent}
In the  last step of the calculation appearing in \S\ref{sec:quotient-by-wprime}, 
we implicitly made use of  a version of descent for $G$-linearized vector bundles equipped with
additional structure, in which $G = W/W'$.  In effect, we considered the family
\[
X_T(\frak B)/\!\!/W' \to X^0(\frak B)\, ,
\]
with fiber $X^1(\frak B)/\!\!/W'$.
The displayed map, induced by projection, is well-defined since $W'$ acts trivially
on $X^0(\frak B)$. Moreover, the family is trivial: $X_T(\frak B)/\!\!/W' \cong  X^0(\frak B) \times (X^1(\frak B)/\!\!/W')$.
 We then produced descent data by constructing
a set of linearly independent $(W/W')$-equivariant sections. 
Unfortunately, this does not quite make sense since the family is not a $(W/W')$-linearized vector bundle.
 However, the family is a subvariety of a trivial $(W/W')$-linearized vector bundle, in which the fiber is
 a subvariety of an affine space.

To justify our calculation, we formulate a version of Proposition \ref{prop:no-name} 
for families of pairs $(E,Z)$,  in which $E$ is a trivial vector bundle on a variety $X$ on which a finite group $G$ acts almost freely, and $Z\subset E$
is an equivariant family of subschemes. 

Let $G$ almost freely
on a variety $X$. After passing to a suitable open subset of $X$, we may assume that
$G$ acts freely and $X\to X/G$ is a $G$-torsor.
 Let $V$ be a  representation of $G$. Then we have a (trivial)  $G$-linearized vector bundle
\[
X \times V \to X
\]
where $X\times V$ has the diagonal action.
Suppose $Z\subset V$ is an equivariant closed subvariety. Then $X \times Z \subset X\times V$ is an equivariant closed
subvariety. 

\begin{prob} Describe the descent data for the pair $(X\times V,X\times Z)$ 
with respect to descending along the torsor $X \to X/G$.
\end{prob}

 We proceed as follows: Let $\Bbb A^n$ be affine $n$-space equipped with
{\it trivial} $G$-action. By descent for vector bundles,
 there is an equivariant isomorphism of vector bundles
\[
X\times \Bbb A^n \cong X\times V\, .
\]
Using the above isomorphism, there is an equivariant closed subvariety 
\[
E \subset X\times \Bbb A^n
\] 
which
corresponds to $X\times Z$ with respect to the above isomorphism.
Then the prime ideal 
\[
I(E) \subset k[X][t_1,\dots,t_n]
\]
which defines $E$ as a vanishing locus
is  $G$-invariant in the sense that elements of $G$ map $I(E)$ to itself. 
Note that $G$ acts trivially on the variables $t_1,\dots,t_n$.

Furthermore the inclusion $X\times Z \subset X \times V$ corresponds under the isomorphism
to an inclusion $E\subset X\times \Bbb A^n$ so we have an equivariant family
\[
Z' \to E\to X
\]
that is isomorphic to $Z \to X \times Z \to X$.

\begin{hypo} \label{hypo:descent} The ideal 
\[
I(E)\subset k[X][t_1,\dots,t_n]
\] is generated by $G$-invariant elements. That is
\[
I(E) \cong I(E)^G \otimes_{k[X]^G[t_1,\dots,t_n]} k[X][t_1,\dots,t_n]\, .
\]
\end{hypo}

\begin{rem}  In our application, 
the hypothesis holds by inspection.
\end{rem}

\begin{lem} The pair 
\[
(X\times V,X\times Z)
\]
 descends along the torsor $X\to X/G$  if and only if Hypothesis \ref{hypo:descent} holds.

Moreover, a rational quotient for $(X\times Z)/\!\!/G$   is given by
\[
(X\times Z)/\!\!/G := \spec (k[X]^G[t_1,\dots,t_n]/I(E)^G) \, .
\]      
\end{lem}

Note that $E/\!\!/G \subset (X/G) \times \Bbb A^n$.  

\begin{proof} It is straightforward to check that the hypothesis is necessary for $E$ to descend.
Conversely, we have $G$-equivariant isomorphisms
\begin{align*}
k[X\times Z] & \cong k[E]\, , \\ &\cong k[X][t_1,\dots,t_n]/I(E) \, ,  \\ &\cong 
(k[X] ^G[t_1,\dots,t_n]/I^G(E)) \otimes_{k[X]^G[t_1,\dots,t_n]} k[X][t_1,\dots,t_n] \, ,
\end{align*}
where the last isomorphism follows from the hypothesis.
Taking $G$-invariants, we obtain an isomorphism
\[
k[X\times Z]^G \cong k^G[X] [t_1,\dots,t_n]/I(E)^G\, . 
\]
Then $\spec$ applied  to the algebra surjection
\[
k[X]^G[t_1,\dots,t_n] \to k[X]^G [t_1,\dots,t_n]/I(E)^G
\]
induces the descended pair $((X/G) \times \Bbb A^n,E/\!\!/G)$.
\end{proof}

\section{X-states on 2-qubits} \label{sec:rationality-for-2-qubits}
The goal of this section is to describe an explicit transcendence basis
for the  field of invariant functions on the variety of X-states on 2-qubits.
However, we wish emphasize that the approach given here will not generalize to $n > 2$. 

In what follows, to make it clear that we
are considering only the 2-qubit case, we will denote the space of L-states on 2-qubits  by
$\scr L_2$ and the variety of $X$-states on 2-qubits by $\scr X_2$.
We also identify the vector spaces  $V_1$ and $V_2$ with $\Bbb C^2$.

\subsection{A different rational section}
In the case of 2-qubits, an L-state $\rho \in \scr L_2$ is given by
\[
(v,w,C)\, ,
\]
in which 
\begin{itemize}
\item $v\in \scr V_1 \cong \Bbb C^3$,
\item $w\in \scr V_2 \cong \Bbb C^3$, and 
\item $C \in \scr V_1 \otimes \scr V_2 \cong \hom(\scr V_1,\scr V_2) \cong M_3(\Bbb C)$. 
\end{itemize}
It will be convenient to display these data as a matrix
\[
{\small\left(
\begin{array}{c|c}
1 & w \\
\hline
v & C
\end{array}
\right)}
\]     
Let $\cal S= \Bbb C \times \Bbb C \times \Bbb C^3$. Then
one has an embedding $\cal S\to \scr L_2$ given by
\begin{equation} \label{eqn:2-qubit-section}
(x,y,\lambda) \mapsto 
{\small \left(
\begin{array}{c|c cc }
\begin{array}{c}
1  \end{array}
& 0 & 0 & y \\
\hline
0 & \lambda_1 & 0  & 0\\
0 & 0& \lambda_2 & 0\\
x & 0 & 0 & \lambda_3
\end{array}
\right)}
\end{equation}
This embedding factors as
\[
\cal S \to \scr X_2 \to \scr L_2\, .
\]
We next describe the normalizer of $\cal S$ in $\scr X_2$.
Let $S_2$ be the permutation group of order two, and let $K_4$ be the Klein 4-group
which we regard as the group of diagonal $2\times 2$ matrices with coefficients in $\pm 1$.
Then $S_2$ acts on $K_4$ by permuting diagonal matrix entries. The semi-direct product
\[
D_8 := K_4\rtimes S_2
\]
is identified with the dihedral group of order 8 and may be interpreted as the group of $2 \times 2$ permutation matrices (aka monomial matrices) with coefficients in $\pm 1$.
Note that $D_8$ embeds canonically in $\SO_3(\Bbb C)$ by 
\[
 D_8 \subset \Or_2(\Bbb C) @>>> \SO_3(\Bbb C)\, ,
\]
where the second homomorphism in the displayed composition is given by
\[
A\mapsto \begin{pmatrix}
A & 0 \\
0 & \det A
\end{pmatrix}\, .
\]
We may therefore regard any element of $D_8$ as an element of $\SO_3(\Bbb C)$.

Set
\[
N := (K_4 \times K_4) \rtimes S_2 \, ,
\]
where $S_2$ acts diagonally on $K_4 \times K_4$. Then $N$ is a group of order 32.
Note that $N \cong D_8 \times_{S_2} D_8$ and in particular, $N$ is a subgroup of
$G = \SO_3(\Bbb C) \times \SO_3(\Bbb C)$. It is straightforward to check that
$N$ is a subgroup of the normalizer of $\cal S \subset \scr X$ with respect to $G$.

\begin{lem}\label{lem:normalizer} The subgroup $N\subset G$ is the normalizer of $\cal S$ in $\scr X$. Moreover,
the centralizer is trivial.
\end{lem}

\begin{proof} Let $g\in G$ satisfy the condition that 
 $g\cdot  (x,y,\lambda)$ lies in $\cal S$ 
for every $(x,y,\lambda) \in \cal S$. We must show show that $g$ lies
in $N$.  It suffices to show this for all $(x,y,\lambda)$
in the Zariski open dense subset of $\cal S$ defined by
the condition
\begin{equation}\label{eqn:condition}
x^2y^2\lambda^2_1\lambda^2_2\lambda^2_3\prod_{i < j} (\lambda^2_i - \lambda^2_j)^2 \ne 0\, .
\end{equation}
If $g = (g_1,g_2)$, then it follows that $g_1$ and $g_2$ preserve the the lines through $(0,0,x) \in \scr V_1$ and $(0,0,y)\in \scr V_2$ respectively. Furthermore,
we necessarily have that $g_2 \lambda g_1^{-1}$ is again a diagonal matrix and this will imply
that the matrices $g_2\lambda^2 g_2^{-1}, g_1\lambda^2 g_1^{-1}$ are both diagonal and of the form
\[
\begin{pmatrix}
 \lambda^2_{\sigma(1)} & 0 &0  \\
0 &   \lambda^2_{\sigma(2)} & 0 \\
0 & 0 &\lambda^2_3
\end{pmatrix}
\]
where $\sigma\in S_2$ is the same for both $g_1$ and $g_2$. From this we infer that $(g_1,g_2)$ lies in $N$. 
Hence, $N$ is the normalizer of $\cal S$.

The non-trivial elements of $N$ act on the matrix appearing in \eqref{eqn:2-qubit-section} by changing the signs
of entries and by possibly permutating the pair $(\lambda_1,\lambda_2)$. Using this description,
it follows that no non-trivial element of $N$ stabilizes a matrix as in \eqref{eqn:2-qubit-section} that satisfies 
the condition \eqref{eqn:condition}. Hence, the centralizer is trivial.
\end{proof}

\begin{prop} $\cal S\subset \scr X_2$ is a rational section.
\end{prop}

\begin{proof} 
We define $\scr X^0_2 \subset \scr X_2$ to be the Zariski open dense subset defined by the set of  X-states
$(v,w,C)\in \scr L_2$ satisfying the condition
\[
\langle v,v\rangle \langle w,w \rangle\det(C^tC)\Delta(C^tC) \ne 0.
\]
Set $\cal S_0 := \scr X^0_2 \cap\cal S$ (these are the the matrices
of \eqref{eqn:condition}). Then $\cal S_0 \subset \cal S$ is open and dense.

Given $(v,w,C) \in \scr X_2^0$, there is $g\in G$ such that $g\cdot (v,w,C)$ lies in $\cal S_0$ by \cite[lem.~4.4]{CVKR-rationality} or \cite[thm.~2]{Choudury-Horn}.
It follows that $G\cdot \cal S_0 = \scr X^0_2$. Taking  the closure of both sides of the latter, we
obtain the first condition of Definition \ref{defn:relativeSection}.

As to the second condition, it suffices to show that if 
\[
(x,y,\lambda) = g\cdot  (x',y',\lambda')
\]
for some $(x,y,\lambda),  (x',y',\lambda') \in \cal S_0$ and $g\in G$, then $g$ lies
in $N$. One establishes this by essentially the same argument that appears in the proof of Lemma \ref{lem:normalizer}.
The  details are left to the reader.
It follows that $\cal S$ is a rational section.
\end{proof}

\begin{cor} The restriction map $\Bbb C(\scr X_2)^{G} \to \Bbb C(\cal S)^{N}$ is an isomorphism.
\end{cor} 

\begin{rem} The normalizer for the section $\cal S$ has dimension zero. 
This relies on the  fact that complex matrices admit singular value decomposition.
This will not generalize in any obvious way to X-states on $n$-qubits for $n > 2$. 
\end{rem}

\begin{conjecture} For $n > 2$, the variety $\scr X_n$ 
with action of $G$ does not admit a section with a discrete normalizer.
\end{conjecture}

\subsection{A transcendence basis}
  If $G$ is a reductive group acting on an affine variety $X$, one says that a collection
   $\{f_1, \ldots, f_n\} \subseteq \mathbb{C}(X)^G$ {\em separates orbits in general position} 
   if there exists a Zariski-open (dense) subset $U\subseteq X$ such that, for every pair of elements $x_1, x_2 \in U$, if $f_i(x_1) = f_i(x_2)$ for all $i=1, \ldots, n$, then $x_1$ and $x_2$ belong to the same $G$-orbit, that is, there exists $g\in G$ such that $gx_1 = x_2$. 
     
Let $G = \SO(\scr V_1) \times \SO(\scr V_2)$. 
Consider the $G$-invariant polynomial functions
\[
p_i\: \scr L_2 \to \Bbb C, \quad i =1,\dots, 5,
\]
in which
\begin{itemize}
 \item $ p_1(v,w,C)= \langle v,v\rangle = \|v\|^2$,
 \item $p_2(v,w,C)=  \langle w,w\rangle = \|w\|^2$,
 \item  $p_3(v,w,C)= \langle Cv,w \rangle$,
 \item  $p_4(v,w,C)= \tr(C^tC)$, and 
 \item $p_5(v,w,C)= \det(C)$.
  \end{itemize}
Restricting the $p_i$ to $\scr X_2$,  we obtain
  five invariants
  \[
  p_i \in \Bbb C[\scr X_2]^G\, .
  \]
  
  \begin{prop} The invariants  $p_1,\dots,p_5$ separate orbits in general position. 
  \end{prop}
  
  \begin{proof}  If we restrict each of the functions $p_i$
  to  $\cal S_0$, 
 we obtain the $N$-invariant functions which assign to a point $\alpha = (x,y,\lambda) \in S$  the expressions
  \begin{enumerate}
 \item $f_1= x^2$,
  \item $f_2 = y^2$,
  \item $f_3=\lambda_3 xy$
  \item $f_4= \lambda^2_1+\lambda^2_2 +\lambda^2_3$, and
  \item $f_5 = \lambda_1 \lambda_2\lambda_3$.
  \end{enumerate}
  Then $\{p_1,\dots,p_5\} \subset \Bbb C(\scr X_2)^G$ separates orbits on $\scr X_0$ if and only if
 $\{f_1,\dots,f_5\} \subset \Bbb C(\cal S)^N$ separates orbits on $\cal S_0 = \scr X^0_2 \cap \cal S$.
  
  Let $\alpha = (x,y,\lambda)$ and  $\alpha' = (x',y',\lambda')$ be points of $\cal S_0$ such that $f_i(\alpha) = f_i(\alpha')$ for $i = 1,\dots 5$.
  Then  the first three invariants $f_1,f_2,f_3$ give relations of the form
  \[
  x =\sigma_1 x',\quad  y = \sigma_2 y',\quad  \lambda_3 = \sigma_3\lambda'_3,
  \]
  where $\sigma_i = \pm 1$ and $\sigma_1\sigma_2\sigma_3 = 1$. But then we can apply 
  a suitable element of $N$
  to arrange it so that $\sigma_i = 1$ for all $i$. Consequently, we may assume that 
  \[
  x =x',\quad  y = y',\quad  \lambda_3 =\lambda'_3,
  \]
  
  Then  $f_4$ implies the relation
  \[
  \lambda^2_1+\lambda^2_2 = (\lambda'_1)^2+(\lambda'_2)^2
  \]
  and $f_4$ and $f_5$ taken together imply, after possibly reordering, that
  \[
  \lambda_1^2 = \tau\lambda'_1, \quad   \lambda_2 = \tau \lambda'_2, 
\]  
where $\tau = \pm 1$. Again, we may apply an element of $N$ to arrange  $\tau = 1$. 
Consequently  the collection $\{f_1,\dots,f_5\}$ separate orbits on $\cal S_0$.
\end{proof}
  
\begin{cor} The invariant functions $p_1,\dots,p_5$ are algebraically independent. Consequently, they induce an isomorphism 
\[
\Bbb C(\scr X_2)^G \cong \Bbb C(p_1,\dots,p_5)\, .
\]
\end{cor}

\begin{proof} By \cite[lem~2.1]{PV}, the invariants $p_1,\cdots, p_5$ generate $\Bbb C(\scr X_2)^G$ since they
separate orbits in general position. However, we already know that
$\Bbb C(\scr X_2)^G$  has transcendence degree 5. Consequently, $p_1,\dots,p_5$ are necessarily algebraically independent, for
otherwise they could not generate $\Bbb C(\scr X_2)^G$.
\end{proof}


\end{document}